\documentclass[11pt]{article}

\usepackage[a4paper,margin=1.12in]{geometry}
\usepackage[T1]{fontenc}
\usepackage[utf8]{inputenc}
\usepackage{lmodern}
\usepackage{amsmath,amssymb,amsthm,mathtools}
\usepackage{microtype}
\usepackage{tikz}
\usepackage[colorlinks=true,linkcolor=black,citecolor=black,urlcolor=black]{hyperref}

\newtheorem{theorem}{Theorem}[section]
\newtheorem{lemma}[theorem]{Lemma}
\newtheorem{proposition}[theorem]{Proposition}
\newtheorem{corollary}[theorem]{Corollary}
\theoremstyle{definition}
\newtheorem{definition}[theorem]{Definition}
\newtheorem{example}[theorem]{Example}
\theoremstyle{remark}
\newtheorem{remark}[theorem]{Remark}

\newcommand{\apart}{\mathrel{\ddagger}}          
\newcommand{\cross}{\mathrel{\between}}          
\newcommand{\sep}{\mathrel{\mid}}                
\newcommand{\dy}[1]{\mathcal{D}(#1)}             
\newcommand{\hull}{\mathsf{K}}                   
\newcommand{\cl}{\mathrm{cl}}
\newcommand{\R}{\mathbb{R}}
\newcommand{\atomf}[2]{\langle #1 \apart #2 \rangle} 
\newcommand{\Lang}{\mathcal{L}}
\newcommand{\gd}{\mathrm{gd}}
\newcommand{\Atoms}{\mathrm{At}}
\newcommand{\supp}{\mathrm{supp}}
\newcommand{\dist}{\mathrm{dist}}
\newcommand{\modsep}{\models_{\mathrm{sep}}}
\newcommand{\modgeo}{\models_{\mathrm{geo}}}
\newcommand{\vdz}{\vdash_{0}}
\newcommand{\marg}{\mu}

\title{\bfseries A Calculus of Apartness over Separoids:\\
Effective Convex Representation, Stratified Conservativity,\\
and the Complexity of Entailment}

\author{Faruk Alpay$^{1,*}$\qquad Baris Ba\c{s}aran$^{1}$\\[6pt]
{\small $^{1}$Department of Computer Engineering, Bahcesehir University, Istanbul, Turkey}\\[2pt]
{\small \texttt{faruk.alpay@bahcesehir.edu.tr}\;\;\texttt{baris.basaran@bahcesehir.edu.tr}}\\[2pt]
{\small $^{*}$Correspondence: \texttt{alpay@lightcap.ai}}}

\date{June 2026}

\begin{document}
\maketitle

\begin{abstract}
\noindent
Every finite family of compact convex bodies in a Euclidean space induces a relation of
\emph{apartness} between disjoint index sets: two index sets are apart when the convex hulls
of the corresponding unions of bodies are disjoint. We take this relation as the primitive of
a finite propositional language and organize the resulting model theory, representation
theory, and consequence problem. Three closure laws (symmetry, bilateral subsumption,
vacuity) axiomatize the relation; in separoid language this is the complement-side, or
separation-polarity, rendering of acyclic separoids. The contribution beyond this
cryptomorphic axiomatics is an effective rational representation theorem with uniform
margins and an exact account of the logical consequences it induces. Every finite apartness
separoid is realized by rational polytopes whose ambient coordinates are indexed by the
maximal separations of the structure. The construction is output-sensitive: maximal
separations and minimal Radon partitions can be enumerated from a full table, generators, or
a membership oracle; the coordinate values have polynomial bit height in the site alphabet;
and every coordinate is a readable certificate of one maximal separation. The realization
carries Euclidean clearance at least $2$ on every bilateral separation, is invariant under
outer parallel enlargement by any radius below $1$, and yields full-dimensional
$C^{1,1}$ bodies after thickening. The distance-function layer is used as convex-analytic
stability bookkeeping: it records Lipschitz comparison, monotonicity under inclusion, and
outer parallel bodies; its eikonal interpretation is contextual rather than an independent
PDE theorem. On the syntactic side, positive entailment is exactly one-premise subsumption.
The Boolean closure is sound, complete, and decidable for consequence over Euclidean scenes;
satisfiability is NP-complete, validity is coNP-complete, and positive entailment is linear
for sorted encodings. A formal stratification condition yields a conservativity theorem: the
Boolean stratum is inert on the atomic stratum, so compound reasoning introduces no new
apartness beyond closure. Finally, the consequence relations of fixed ambient dimension form
a strictly decreasing hierarchy that stabilizes at dimension exactly $|E|-1$, where $E$ is
the site alphabet; the stabilization threshold uses the Strausz-Bracho representation theorem,
while the proof-readable rational-margin representation is independent and generally higher
dimensional.
\medskip


\end{abstract}

\section{Introduction}\label{sec:intro}

A hyperplane that leaves one family of convex bodies on one side and another family on the
other side certifies a single bit of information: the two families are \emph{apart}. The
present paper is a study of what can be said, proved, and decided when these bits are all that
is retained. We fix a finite set $E$ of \emph{sites}, attach to each site a nonempty compact
convex body in some Euclidean space, and record, for every pair of disjoint subsets
$A,B\subseteq E$, whether the convex hulls of the corresponding unions of bodies are disjoint.
The record is a relation $\apart$ on disjoint pairs. Everything else about the bodies (their
dimension, their shape, their coordinates) is deliberately forgotten, and one of the paper's
aims is to determine exactly how much the forgetting costs. The answer, made precise in
Theorem~\ref{thm:hierarchy}, is that it costs nothing once the ambient dimension passes a
threshold that we compute exactly: each dimension below $|E|-1$ is genuinely new, in that it
strictly enlarges the stock of refutable assertions, while the dimensions from $|E|-1$ onward
repeat one another without leaving any trace in the consequence relation.

The combinatorial structures that arise in this way are not new. They are the
\emph{separoids} of Strausz and his collaborators
\cite{ABMOS02,Strausz03,BS06,NS06,Strausz07}, introduced in geometric transversal theory as
the common abstraction of Radon partitions of point configurations, oriented matroids, and
separation patterns of convex bodies. The separoid literature is resolutely combinatorial and
categorical: it studies homomorphisms, universality, Tverberg-type colourings, and the
topology of transversal spaces \cite{MPS06,MS04}. The present paper adds a logical and effective layer to that theory: a formal language whose
atoms are separation assertions, a syntactic closure calculus for the positive fragment, a
Boolean consequence theorem for Euclidean scenes, and a complexity classification of the
associated decision problems. The axiomatics themselves are not presented as new; they are
the separoid axioms seen from the complementary side. The new representation result is the
quantitative one needed for that logical reading. It is proved from first principles with
rational coordinates, explicit coordinate certificates, and uniform margins. A single
external theorem is imported from the literature, the dimension-$(|E|-1)$ realization of
\cite{Strausz03,BS06}; it is used only to identify the exact stabilization threshold of the
fixed-dimension hierarchy.

\subsection*{Contributions}

Throughout, $E$ is a finite set and $\dy{E}$ denotes the set of pairs $(A,B)$ of disjoint
subsets of $E$.

\paragraph{Axiomatics and cryptomorphy.} Section~\ref{sec:scenes} verifies that the apartness
relation of every family of compact convex bodies satisfies three laws: symmetry, bilateral
subsumption (apartness passes to componentwise subsets), and vacuity (the empty set is apart
from everything). Section~\ref{sec:separoids} takes these laws as axioms; the resulting
structures, which we call \emph{apartness separoids}, are shown to be cryptomorphic to the
acyclic separoids of \cite{Strausz03}, the translation being the involutive exchange of a
separation relation with its complementary relation of Radon partitions. This paragraph is a
placement statement, not a novelty claim: symmetry, monotonicity, and vacuity are inherited
from separoid theory after polarity is reversed. The word \emph{cryptomorphic} is used here
exactly as in matroid theory \cite{BLVSWZ}: two axiom systems, each complete, whose models
determine one another by an explicit dictionary.

\paragraph{Effective representation with margins.} Section~\ref{sec:representation} proves
that the three laws are complete for Euclidean realizability: every apartness separoid is the
apartness relation of a family of rational polytopes (Theorem~\ref{thm:representation}). The
construction assigns one coordinate to each maximal separation and one witness point to each
minimal Radon partition; all coordinates are rationals of explicitly bounded height. The
algorithmic meaning of \emph{effective} is made explicit in Proposition~\ref{prop:effective}:
from a full table, a generator presentation, a minimal-Radon presentation, or a membership
oracle, the maximal and minimal objects can be enumerated with the expected exponential
output cost and without hidden bit growth. Proposition~\ref{prop:position} then separates
what is inherited from \cite{Strausz03,BS06}, what is a quantitative strengthening, and what
is a re-indexed proof device. The realization is quantitatively robust: every separation
holds with Euclidean margin at least $2$, so the relation is unchanged when every body is
replaced by its outer parallel body at any radius $r<1$, and is in particular realized by
full-dimensional bodies with $C^{1,1}$ boundaries (Theorem~\ref{thm:stability}). The margin
is organized through distance functions and outer parallel bodies; the eikonal reading is
recorded as context for the same comparison facts, not as a separate PDE contribution
(Lemmas~\ref{lem:distance} and~\ref{lem:threshold}). Two exact computations close the
section: the geometric dimension of the $d$-dimensional simploid equals $d$
(Proposition~\ref{prop:simploid}), and the fixed-dimension consequence relations form a chain
\[
\models_0\ \supsetneq\ \models_1\ \supsetneq\ \cdots\ \supsetneq\ \models_{|E|-2}
\ \supsetneq\ \models_{|E|-1}\ =\ \models_{|E|}\ =\ \cdots\ =\ \modgeo,
\]
strict at every step below the threshold and constant from it onward
(Theorem~\ref{thm:hierarchy}).

\paragraph{Logical consequence.} Section~\ref{sec:calculus} introduces the language
$\Lang_E$, whose atoms are the assertions $\atomf{A}{B}$ for $(A,B)\in\dy{E}$, and a four-rule
sequent calculus for its positive fragment. Positive entailment collapses to subsumption
(Theorem~\ref{thm:subsumption}): an apartness assertion follows from a set of apartness
assertions precisely when it is vacuous or componentwise dominated, possibly after a swap, by
a single premise. This is framed as the exact syntactic shadow of separoid closure, not as a
large proof calculus: every derivable sequent has a derivation of height at most three, cut
is admissible, and interpolation trivializes. The Boolean closure is sound and complete for
consequence over scenes and decidable (Theorem~\ref{thm:completeness}); we separate the
routine finite propositional completeness from the geometric representation content.

\paragraph{Complexity.} Section~\ref{sec:complexity} classifies the decision problems:
satisfiability of a formula of $\Lang_E$ is NP-complete, validity is coNP-complete, and
positive entailment is decidable in time linear in the input. The lower bound rests on the
observation that singleton-pair atoms over disjoint site alphabets form an independent
family: every truth assignment to them is realized by a structure, hence by a scene of
rational polytopes.

\paragraph{Stratification and conservativity.} Section~\ref{sec:stratification} isolates a
discipline on presentations (every scheme constrains only material of strictly lower grade,
and no scheme exhibits a closed instance of the grade it derives), verifies that the calculus
obeys it, and proves the corresponding conservativity theorem: the Boolean stratum is inert
on the atomic stratum (Theorem~\ref{thm:conservativity}). No amount of compound reasoning,
classical negation included, ever forces an apartness that was not already a subsumption.
Together with Theorem~\ref{thm:stability} this gives the two halves of one fact: small
perturbations of the bodies cannot change the relation, and no reasoning over the relation
can change it either.

\subsection*{Related work}

Separoids were introduced in \cite{ABMOS02} and developed in
\cite{Strausz03,BS06,NS06,MPS06,MS04,Strausz07}; they generalize oriented matroids
\cite{BLVSWZ} and abstract the Radon partitions of \cite{Radon21}. The axiomatic study of
convexity through its separation and exchange behaviour goes back at least to Levi
\cite{Levi51} and continues through convex geometries \cite{EJ85} and the general theory of
convex structures \cite{vdV93}. The word \emph{apartness} and the discipline of taking
apartness rather than nearness as primitive are borrowed from constructive topology
\cite{BV11}; the borrowing is terminological, and Remark~\ref{rem:classical} delimits it
precisely. The classical counterpoint is the theory of proximity spaces \cite{NW70}, from
which the present relation differs in one law that changes everything
(Remark~\ref{rem:proximity}). Region-based spatial logics in the tradition of \cite{RCC92},
surveyed in \cite{Handbook07}, also reason about qualitative relations between extended
regions; the relation studied here is not among their primitives, because it is evaluated on
convex hulls of unions rather than on unions. The stability material of
Section~\ref{sec:margins} uses classical facts about distance functions and outer parallel
bodies, with the eikonal interpretation recorded only as background context
\cite{CL83,BCD97,Schneider14,Federer59}. Finally, Section~\ref{sec:conclusion} contrasts the
tameness established here with the universality phenomena that govern realizability of point
configurations in fixed dimension \cite{Mnev88}.

\section{Scenes and the apartness of convex bodies}\label{sec:scenes}

Throughout the paper $E$ is a finite nonempty set whose elements are called \emph{sites}, and
\[
\dy{E} \;=\; \{(A,B) : A,B\subseteq E,\ A\cap B=\emptyset\}
\]
is the set of \emph{disjoint pairs} over $E$. For $(A,B),(A',B')\in\dy{E}$ we write
$(A,B)\sqsubseteq (A',B')$ when $A\subseteq A'$ and $B\subseteq B'$, and we say that
$(A',B')$ \emph{dominates} $(A,B)$. A pair is \emph{vacuous} if one of its components is
empty, and \emph{bilateral} otherwise.

\begin{definition}[Scenes]\label{def:scene}
A \emph{scene} over $E$ is a pair $\mathcal{C}=(d,(C_e)_{e\in E})$ where $d\geq 0$ and each
$C_e\subseteq\R^d$ is a nonempty compact convex set. For $A\subseteq E$ put
\[
\hull_{\mathcal C}(A)\;=\;\operatorname{conv}\Bigl(\bigcup_{a\in A} C_a\Bigr),
\qquad \hull_{\mathcal C}(\emptyset)=\emptyset .
\]
The \emph{apartness relation of $\mathcal C$} is
\[
{\apart_{\mathcal C}} \;=\; \bigl\{(A,B)\in\dy{E} :
\hull_{\mathcal C}(A)\cap \hull_{\mathcal C}(B)=\emptyset\bigr\}.
\]
We read $A\apart_{\mathcal C} B$ as ``$A$ is apart from $B$ in $\mathcal C$''. The
complementary relation on $\dy{E}$ is written $\cross_{\mathcal C}$ and read ``$A$ crosses
$B$''; following \cite{Radon21,Strausz03}, a crossing pair is also called a \emph{Radon
partition} of $\mathcal C$.
\end{definition}

Since each $C_e$ is compact and $E$ is finite, $\bigcup_{a\in A}C_a$ is compact, and the
convex hull of a compact subset of $\R^d$ is compact; hence each $\hull_{\mathcal C}(A)$ with
$A\neq\emptyset$ is a nonempty compact convex set. The name \emph{apartness} is justified, and
quantified, by the following standard fact, recorded with proof because its margin is reused
throughout Section~\ref{sec:representation}.

\begin{definition}[Margin]\label{def:margin}
For compact convex $K,L\subseteq\R^d$ set $\marg(K,L)=\dist(K,L)=
\min\{\lVert x-y\rVert : x\in K,\ y\in L\}$ if both are nonempty, and
$\marg(K,L)=+\infty$ otherwise. For a scene $\mathcal C$ and $(A,B)\in\dy{E}$ write
$\marg_{\mathcal C}(A,B)=\marg\bigl(\hull_{\mathcal C}(A),\hull_{\mathcal C}(B)\bigr)$.
\end{definition}

\begin{lemma}[Strict separation, with margin]\label{lem:separation}
Let $K,L\subseteq\R^d$ be compact convex sets. Then $K\cap L=\emptyset$ if and only if
$\marg(K,L)>0$, if and only if there exist $u\in\R^d$, $c\in\R$, and $\varepsilon>0$ such
that $\langle u,x\rangle\le c-\varepsilon$ for all $x\in K$ and
$\langle u,y\rangle\ge c+\varepsilon$ for all $y\in L$. When both sets are nonempty one may
take $\varepsilon=\marg(K,L)^2/2$ with $\lVert u\rVert=\marg(K,L)$.
\end{lemma}

\begin{proof}
If such $u,c,\varepsilon$ exist, no point lies in both sets. If a set is empty, both
disjointness and $\marg=+\infty$ hold and the linear conditions on the empty side are
vacuous; take $u=0$, $c=-1$, $\varepsilon=\tfrac12$ when $K=\emptyset$. So assume both
nonempty and disjoint. The function $(x,y)\mapsto \lVert x-y\rVert$ is continuous on the
compact set $K\times L$, hence attains its minimum $\delta=\marg(K,L)$ at some $(p,q)$, and
$\delta>0$ by disjointness. Put $u=q-p$. We claim
$\langle u,x\rangle\le\langle u,p\rangle$ for all $x\in K$: otherwise
$\langle q-p,\,x-p\rangle>0$ for some $x\in K$, and with $p_t=p+t(x-p)\in K$ for
$t\in[0,1]$ we get
$\frac{d}{dt}\lVert p_t-q\rVert^2\big|_{t=0}=2\langle p-q,\,x-p\rangle<0$, so
$\lVert p_t-q\rVert<\delta$ for small $t>0$, contradicting minimality. Symmetrically
$\langle u,y\rangle\ge\langle u,q\rangle$ for all $y\in L$. Since
$\langle u,q\rangle-\langle u,p\rangle=\lVert u\rVert^2=\delta^2$, the choice
$c=\langle u,\tfrac{p+q}{2}\rangle$ and $\varepsilon=\delta^2/2$ works.
\end{proof}

Thus, for scenes, ``the hulls are disjoint'', ``the margin is positive'', and ``some
hyperplane separates the hulls strictly, with quantified clearance'' are one condition, and
this is the strict reading used for separoids of convex sets in \cite{Strausz03}. Bodies that
merely touch are \emph{not} apart, and their margin is $0$.

\begin{proposition}[The three laws]\label{prop:threelaws}
For every scene $\mathcal C$ over $E$, the relation $\apart_{\mathcal C}$ satisfies:
\begin{itemize}
\item[\textup{(A1)}] \emph{Symmetry:} if $A\apart B$ then $B\apart A$.
\item[\textup{(A2)}] \emph{Bilateral subsumption:} if $A\apart B$, $A'\subseteq A$, and
$B'\subseteq B$, then $A'\apart B'$.
\item[\textup{(A3)}] \emph{Vacuity:} $\emptyset\apart B$ for every $B\subseteq E$.
\end{itemize}
\end{proposition}

\begin{proof}
(A1) is the symmetry of intersection. For (A2), $A'\subseteq A$ gives
$\hull(A')\subseteq\hull(A)$ and likewise for $B$, so
$\hull(A')\cap\hull(B')\subseteq\hull(A)\cap\hull(B)=\emptyset$. For (A3),
$\hull(\emptyset)=\emptyset$ meets nothing.
\end{proof}

Note that within $\dy{E}$ the only pair of the form $(A,A)$ is $(\emptyset,\emptyset)$, which
is apart by (A3); this is the quasi-antireflexivity listed for separation relations in
\cite{Strausz03}, here absorbed into the choice of domain. A quantitative refinement of
Proposition~\ref{prop:threelaws}, in which each law becomes a monotonicity statement about
$\marg$, is given in Lemma~\ref{lem:threshold}.

\begin{remark}[Against additivity]\label{rem:proximity}
A proximity relation $\delta$ in the sense of \cite{NW70} satisfies the additivity law
$A\,\delta\,(B\cup C)\iff A\,\delta\,B$ or $A\,\delta\,C$, and the connection relations of
\cite{RCC92} behave likewise on unions. The crossing relation of a scene violates additivity
in the smallest possible configuration: take $d=1$ and singleton bodies $C_u=\{0\}$,
$C_v=\{1\}$, $C_w=\{2\}$. Then $\{v\}\apart\{u\}$ and $\{v\}\apart\{w\}$, yet
$\hull(\{u,w\})=[0,2]\ni 1$, so $\{v\}\cross\{u,w\}$. The culprit is the hull: the relation is
evaluated on $\operatorname{conv}(C_u\cup C_w)$, not on $C_u\cup C_w$. This single failure is
what separates the present theory from proximity and connection calculi, and it is the entire
source of logical content below: were additivity available, every relation would be
determined by its singleton pairs and the language of Section~\ref{sec:calculus} would
collapse.
\end{remark}

\begin{example}[A running scene]\label{ex:running}
Let $E=\{a,b,c\}$ and take, in $\R^2$,
\[
C_c=[-1,-\tfrac14]\times[0,1],\qquad
C_a=[0,1]\times[0,1],\qquad
C_b=[4,5]\times[0,1].
\]
All three bodies are pairwise apart: vertical lines $x=-\tfrac18$ and $x=\tfrac52$ provide
the clearance. Moreover $\{b\}\apart\{a,c\}$, witnessed by $x=\tfrac52$ again, and
$\{c\}\apart\{a,b\}$, witnessed by $x=-\tfrac18$. The remaining bilateral pair behaves
differently: $\hull(\{b,c\})=[-1,5]\times[0,1]\supseteq C_a$, so $\{a\}\cross\{b,c\}$, even
though $a$ is apart from $b$ and from $c$ separately. Figure~\ref{fig:running} shows the
scene; the joint hull of $C_b$ and $C_c$ simply sweeps over $C_a$. The crossing pair
$(\{a\},\{b,c\})$ is minimal: every pair strictly below it in $\sqsubseteq$ is apart. The
maximal bilateral separations are $(\{b\},\{a,c\})$ and $(\{c\},\{a,b\})$, and every
separation of the scene is dominated by one of them or is vacuous. The same abstract relation
is realized on a line by the points $-\tfrac12,\ \tfrac12,\ \tfrac92$ for $c,a,b$: it is the
separation structure of three collinear points, one of the eight isomorphism types of acyclic
separoids of order three catalogued in \cite{Strausz03}. A scene is remembered by its
relation; the relation, as Lemma~\ref{lem:generation} will show, is remembered by its maximal
separations and minimal Radon partitions; everything between them is reconstruction.
\end{example}

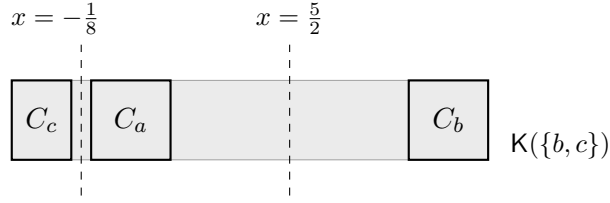
\begin{figure}[t]
\centering
\begin{tikzpicture}[scale=1.05]
\fill[black!8] (-1,0) rectangle (5,1);
\draw[black!35,thin] (-1,0) rectangle (5,1);
\draw[thick] (-1,0) rectangle (-0.25,1);
\draw[thick] (0,0) rectangle (1,1);
\draw[thick] (4,0) rectangle (5,1);
\node at (-0.625,0.5) {$C_c$};
\node at (0.5,0.5) {$C_a$};
\node at (4.5,0.5) {$C_b$};
\draw[dashed] (2.5,-0.45) -- (2.5,1.45);
\node[anchor=south] at (2.5,1.45) {\small $x=\tfrac52$};
\draw[dashed] (-0.125,-0.45) -- (-0.125,1.45);
\node[anchor=south] at (-0.45,1.45) {\small $x=-\tfrac18$};
\node[anchor=west] at (5.15,0.18) {\small $\hull(\{b,c\})$};
\end{tikzpicture}
\caption{The running scene of Example~\ref{ex:running}. The shaded region is the joint hull
$\hull(\{b,c\})$, which swallows $C_a$: hence $\{a\}\cross\{b,c\}$ although $a$ is apart from
each of $b,c$ on its own. The dashed lines witness the two maximal bilateral separations.}
\label{fig:running}
\end{figure}

\section{Apartness separoids}\label{sec:separoids}

We now take the conclusions of Proposition~\ref{prop:threelaws} as axioms.

\begin{definition}[Apartness separoid]\label{def:as}
An \emph{apartness separoid} over $E$ is a pair $\Sigma=(E,\apart)$ with
$\apart\subseteq\dy{E}$ satisfying (A1)--(A3) of Proposition~\ref{prop:threelaws}. Its
\emph{crossing relation} is $\cross_\Sigma=\dy{E}\setminus\apart$; crossing pairs are called
the \emph{Radon partitions} of $\Sigma$, and the support of a Radon partition $(A,B)$ is
$\supp(A,B)=A\cup B$. A Radon partition is \emph{minimal} if it is $\sqsubseteq$-minimal
among Radon partitions. A bilateral pair in $\apart$ is a \emph{bilateral separation}; it is
\emph{maximal} if it is $\sqsubseteq$-maximal among bilateral separations. By (A1), both
notions are invariant under swapping components, and we count maximal separations and minimal
Radon partitions as unordered pairs $\{A,B\}$.
\end{definition}

By (A3) and (A1), every Radon partition is bilateral. Since $E$ is finite, every Radon
partition dominates a minimal one and every bilateral separation is dominated by a maximal
one.

The structures of Definition~\ref{def:as} are exactly the acyclic separoids of the
literature, presented from the other side of the mirror. Recall from \cite{Strausz03} that a
\emph{separoid} is a relation $\dagger\subseteq 2^S\times 2^S$ such that, for all
$A,B\subseteq S$:
$(\circ)$ $A\dagger B\Rightarrow B\dagger A$;
$(\circ\circ)$ $A\dagger B\Rightarrow A\cap B=\emptyset$; and
$(\circ\circ\circ)$ $A\dagger B$ and $C\subseteq S\setminus A$ imply $A\dagger(B\cup C)$.
The separoid is \emph{acyclic} when $\emptyset$ is separated from $S$, that is,
$\neg(\emptyset\dagger S)$.

\begin{proposition}[Cryptomorphy]\label{prop:crypto}
The assignment $\apart\;\longmapsto\;\cross=\dy{E}\setminus\apart$ is a bijection between
apartness separoids over $E$ and acyclic separoids on $E$; its inverse is
$\dagger\;\longmapsto\;\dy{E}\setminus\dagger$. The two axiom systems are therefore
cryptomorphic in the sense in which the word is used for matroid axiomatics \cite{BLVSWZ}.
\end{proposition}

\begin{proof}
Let $\apart$ satisfy (A1)--(A3) and put $\dagger=\dy{E}\setminus\apart$, regarded as a
relation on $2^E\times 2^E$ that holds only on disjoint pairs; then $(\circ\circ)$ holds by
fiat and $(\circ)$ follows from (A1). For $(\circ\circ\circ)$, suppose $A\dagger B$ and
$C\subseteq E\setminus A$. The pair $(A,B\cup C)$ lies in $\dy{E}$, and
$(A,B)\sqsubseteq(A,B\cup C)$; if $(A,B\cup C)$ were in $\apart$ then so would $(A,B)$ be, by
(A2), contradicting $A\dagger B$. Hence $A\dagger(B\cup C)$. Acyclicity is (A3) at $B=E$.

Conversely let $\dagger$ be an acyclic separoid and put
$\apart=\dy{E}\setminus\dagger$. (A1) follows from $(\circ)$. For (A2), suppose
$(A,B)\in\apart$, $A'\subseteq A$, $B'\subseteq B$, and, towards a contradiction,
$A'\dagger B'$. Two applications of $(\circ\circ\circ)$ interleaved with $(\circ)$ climb back
up: from $A'\dagger B'$ and $B\setminus B'\subseteq E\setminus A'$ we get $A'\dagger B$; from
$B\dagger A'$ and $A\setminus A'\subseteq E\setminus B$ we get $B\dagger A$, that is,
$A\dagger B$, a contradiction. For (A3), if $\emptyset\dagger B$ for some $B$, then
$(\circ\circ\circ)$ with $C=E\setminus B$ gives $\emptyset\dagger E$, contradicting
acyclicity. The two assignments are mutually inverse complementations.
\end{proof}

In the separoid literature the separation relation is written $A\sep B$ and the structure is
denoted $(S,\sep)$ or $(S,\dagger)$ interchangeably \cite{Strausz03}; we keep the symbol
$\apart$ and the apartness reading, after the constructive usage of \cite{BV11}, because the
development below treats the relation as the primitive and never reconstitutes points.

\begin{remark}[Scope of the borrowing]\label{rem:classical}
The borrowing from \cite{BV11} is terminological. The object theory here is finite and the
metatheory classical; nothing below depends on, or contributes to, constructive apartness
spaces. The word is kept because it names the primitive accurately: the relation asserts
positive, witnessed separation (Lemma~\ref{lem:separation} supplies the witness with a
margin), not the mere negation of contact.
\end{remark}

\begin{definition}[Closure]\label{def:closure}
For $\Gamma\subseteq\dy{E}$ define
\begin{multline*}
\cl(\Gamma)=\bigl\{(A,B)\in\dy{E} \;:\; A=\emptyset\ \text{or}\ B=\emptyset\ \text{or}\\
\exists (A',B')\in\Gamma\ \bigl[(A,B)\sqsubseteq(A',B')\ \text{or}\
(A,B)\sqsubseteq(B',A')\bigr]\bigr\}.
\end{multline*}
\end{definition}

\begin{lemma}[Generation]\label{lem:generation}
For every $\Gamma\subseteq\dy{E}$, the relation $\cl(\Gamma)$ is the least apartness separoid
containing $\Gamma$. Moreover, every apartness separoid $\Sigma$ satisfies
$\apart_\Sigma=\cl(M_\Sigma)$, where $M_\Sigma$ is the set of maximal bilateral separations
of $\Sigma$ (one ordered representative per unordered pair suffices); dually, $\Sigma$ is
determined by its minimal Radon partitions, as observed in \cite{Strausz03}.
\end{lemma}

\begin{proof}
$\cl(\Gamma)$ satisfies (A3) by the vacuity clause and (A1) because the defining disjunction
is invariant under swapping $(A,B)$, the two domination cases exchanging places. For (A2), if
$(A,B)\in\cl(\Gamma)$ via a witness $(A',B')$ and $(A'',B'')\sqsubseteq(A,B)$, the same
witness dominates $(A'',B'')$ by transitivity of inclusion; vacuous pairs stay vacuous under
$\sqsubseteq$-descent. Containment $\Gamma\subseteq\cl(\Gamma)$ is the case
$(A',B')=(A,B)$. If $\Sigma$ is any apartness separoid with
$\Gamma\subseteq\apart_\Sigma$, then each $(A,B)\in\cl(\Gamma)$ lies in $\apart_\Sigma$:
vacuous pairs by (A3), dominated pairs by (A2), swapped dominations additionally by (A1).
Hence $\cl(\Gamma)$ is least.

For the second claim, $\cl(M_\Sigma)\subseteq\apart_\Sigma$ by leastness, and conversely
every vacuous pair is in $\cl(M_\Sigma)$ while every bilateral separation is dominated by a
maximal one. Determination by minimal Radon partitions is the complementary statement under
Proposition~\ref{prop:crypto}.
\end{proof}

\begin{definition}[Finite presentations]\label{def:presentations}
Let $N_E=|\dy{E}|=3^{|E|}$. We use four finite access models.
\begin{itemize}
\item A \emph{full table} lists the truth value of $A\apart B$ for all $(A,B)\in\dy{E}$.
\item A \emph{positive generator presentation} is a finite list $G\subseteq\dy{E}$ and
represents the separoid $\cl(G)$.
\item A \emph{minimal-Radon presentation} is a finite symmetric list $R_0$ of bilateral
pairs and represents the structure whose crossing relation is the upward closure of $R_0$:
$(A,B)$ crosses exactly when some member of $R_0$, possibly swapped, is dominated by
$(A,B)$.
\item A \emph{membership oracle} answers whether a queried disjoint pair is apart.
\end{itemize}
The second and third presentations are exact only when the relation they define satisfies
(A1)--(A3); this condition is checkable by the enumeration bounds below.
\end{definition}

\begin{proposition}[Enumeration and bit complexity]\label{prop:effective}
Let $n=|E|$ and $N_E=3^n$. From any of the access models in
Definition~\ref{def:presentations} one can enumerate the maximal bilateral separations
$M_\Sigma$ and the minimal Radon partitions $R_\Sigma$. More precisely:
\begin{itemize}
\item from a full table, or from a membership oracle, $O(N_E^2 n)$ set-comparison work and
$O(N_E)$ table entries or oracle calls suffice after the disjoint pairs have been listed;
\item from a positive generator presentation $G$, membership in $\cl(G)$ costs
$O(|G|n)$, and enumeration costs $O(N_E^2 n+N_E|G|n)$;
\item from a minimal-Radon presentation $R_0$, membership in the represented relation costs
$O(|R_0|n)$, and enumeration costs $O(N_E^2 n+N_E|R_0|n)$;
\item once $M_\Sigma$ and $R_\Sigma$ are known, Theorem~\ref{thm:representation} writes at
most $n(1+|R_\Sigma|)$ generating points in $|M_\Sigma|$ coordinates, each rational having
numerator and denominator of $O(\log n)$ bits.
\end{itemize}
The construction is therefore exponential in $n$ only through the number of possible
disjoint pairs and through its own output dimension; no additional arithmetic blowup is
hidden in the word \emph{effective}.
\end{proposition}

\begin{proof}
Enumerate a disjoint pair by assigning each element of $E$ one of three states: left, right,
or absent. This gives $N_E=3^n$ ordered pairs. A bilateral separation is maximal precisely
when it is apart, both sides are nonempty, and no strictly larger bilateral disjoint pair is
apart. A Radon partition is minimal precisely when it is crossing and no strictly smaller
bilateral disjoint pair is crossing. Testing these two conditions by pairwise comparison
uses $O(N_E^2)$ dominance tests, each implemented by subset tests on bit vectors of length
$n$.

For a full table or oracle, the truth value of each pair is read or queried once and then
cached. For a positive generator presentation, membership is exactly the domination test in
Definition~\ref{def:closure}: vacuity, or domination by a generator, possibly after swap. For
a minimal-Radon presentation, a nonvacuous pair is crossing exactly when it dominates one of
the listed pairs, possibly after swap. The same maximal/minimal scan then applies. The point
count and bit bound are the statement of Theorem~\ref{thm:representation}: every body has one
base point and at most one anchor for each minimal Radon partition containing its site, and
all scalar assignments have denominators at most $n$ and numerators bounded by $n(2n+1)$.
\end{proof}

\begin{remark}[Orientation]\label{rem:orientation}
The unordered maximal separations and minimal Radon partitions can be oriented canonically,
for instance by lexicographic order on the two bit vectors that encode their components. A
different orientation of a maximal separation only reflects the corresponding coordinate;
a different orientation of a minimal Radon partition only swaps the two averages in the same
witness equation. The bit heights and margins of Theorem~\ref{thm:representation} are
unchanged.
\end{remark}

\begin{definition}[The language and the consequence relations]\label{def:language}
The language $\Lang_E$ has one propositional atom $\atomf{A}{B}$ for each
$(A,B)\in\dy{E}$ and is closed under $\neg,\wedge,\vee,\rightarrow$; $\Atoms_E$ is the set of
atoms. An apartness separoid $\Sigma$ evaluates atoms by
$\Sigma\models\atomf{A}{B}\iff(A,B)\in\apart_\Sigma$, and Boolean connectives classically; a
scene $\mathcal C$ evaluates atoms through $\apart_{\mathcal C}$. For
$\Gamma\cup\{\varphi\}\subseteq\Lang_E$ write $\Gamma\modsep\varphi$ when every apartness
separoid over $E$ satisfying $\Gamma$ satisfies $\varphi$, and $\Gamma\modgeo\varphi$ when
every scene over $E$, of every dimension, satisfying $\Gamma$ satisfies $\varphi$. For
$d\geq 0$, $\Gamma\models_d\varphi$ restricts the scenes to ambient space $\R^d$.
\end{definition}

By Proposition~\ref{prop:threelaws}, every scene-induced relation is an apartness separoid;
hence $\modsep\,\subseteq\,\modgeo$, with equality established in
Corollary~\ref{cor:invariance}.

\section{The effective representation theorem}\label{sec:representation}

\begin{definition}\label{def:gd}
An apartness separoid $\Sigma$ is \emph{realized} by a scene $\mathcal C$ when
$\apart_{\mathcal C}=\apart_\Sigma$. The \emph{geometric dimension} $\gd(\Sigma)$ is the
least $d$ such that some scene in $\R^d$ realizes $\Sigma$; the terminology follows
\cite{Strausz03}.
\end{definition}

\subsection{The construction}

\begin{lemma}[One-coordinate averaging]\label{lem:onecoordinate}
Fix a nonempty finite set $X$ of size at most $n$. Suppose each $x\in X$ is assigned one of
three side types: negative, positive, or free. A negative value must lie in $(-\infty,-1]$,
a positive value in $[1,\infty)$, and a free value in $\R$. Let $I(X)$ be the set of averages
of admissible assignments on $X$. Then
\[
I(X)=(-\infty,-1]\quad\text{if all members are negative},
\]
\[
I(X)=[1,\infty)\quad\text{if all members are positive},
\]
and $I(X)=\R$ in every remaining case. Moreover, for every target
$\tau\in\{-1,0,1\}\cap I(X)$ there is an admissible assignment with denominators at most
$n$ and absolute values at most $2n+1$.
\end{lemma}

\begin{proof}
The pure cases are immediate: averaging numbers bounded above by $-1$ gives an average
bounded above by $-1$, and setting all values to the target realizes every target in that
ray; the positive case is symmetric.

Assume $X$ is not pure. If a free member exists, set all negative constrained values to
$-1$, all positive constrained values to $1$, and all free values equal to a scalar $z$. If
there are $r$ negative, $s$ positive, and $f\ge1$ free members, the equation
\[
\frac{-r+s+fz}{|X|}=\tau
\]
has the solution $z=(\tau |X|+r-s)/f$. For $\tau\in\{-1,0,1\}$ the numerator has absolute
value at most $2|X|$, the denominator is at most $|X|$, and the bound $|z|\le 2n$ follows.

It remains to treat the mixed case without free members. Then $r,s\ge1$ and $r+s=|X|$.
Set all negative values to $-1-v$ and all positive values to $1+u$, with $u,v\ge0$. The
average is $\tau$ exactly when
\[
su-rv=\tau |X|+r-s .
\]
If the right side is nonnegative, take $u=(\tau |X|+r-s)/s$ and $v=0$; otherwise take
$u=0$ and $v=(s-r-\tau |X|)/r$. This realizes every real target, hence $I(X)=\R$; for
$\tau\in\{-1,0,1\}$ the same numerator bound gives denominators at most $n$ and absolute
values at most $2n+1$.
\end{proof}

\begin{theorem}[Effective representation]\label{thm:representation}
Let $\Sigma=(E,\apart)$ be an apartness separoid, $n=|E|$. Let $M=\{M_1,\dots,M_k\}$ be its
maximal bilateral separations and $R=\{P_1,\dots,P_m\}$ its minimal Radon partitions. Then
$\Sigma$ is realized in $\R^k$ by a scene of polytopes in which:
\begin{itemize}
\item each body $C_e$ is the convex hull of $1+\#\{P\in R: e\in\supp(P)\}$ explicitly given
points;
\item every coordinate of every point is a rational number $p/q$ with $|p|\le n(2n+1)$ and
$1\le q\le n$;
\item every bilateral separation of $\Sigma$ holds in the scene with margin at least $2$:
$\marg_{\mathcal C}(A,B)\ge 2$ whenever $(A,B)\in\apart$ is bilateral.
\end{itemize}
In particular $\gd(\Sigma)\le k$, and every apartness separoid is realizable by rational
polytopes.
\end{theorem}

\begin{proof}
If $k=0$ there are no bilateral separations at all, so by Lemma~\ref{lem:generation}
$\apart=\cl(\emptyset)$ consists exactly of the vacuous pairs. The scene in $\R^0=\{0\}$ with
every $C_e=\{0\}$ realizes this: every bilateral pair of nonempty hulls meets at $0$, vacuous
pairs are apart, and the margin condition is vacuous. Assume henceforth $k\geq 1$, and fix
for each $i\le k$ an ordered representative $M_i=(A_i^*,B_i^*)$ of the $i$-th maximal
separation; by (A1) the choice of orientation is immaterial, a reflection of the $i$-th
coordinate translating between the two choices.

\medskip\noindent\emph{The points.}
For each minimal Radon partition $P=(A',B')\in R$ (an ordered representative is fixed once)
and each $e\in\supp(P)$ introduce an \emph{anchor} $q_{P,e}\in\R^k$, and for each $e\in E$ a
\emph{base point} $t_e\in\R^k$; the coordinates are fixed below. Put
\[
C_e \;=\; \operatorname{conv}\Bigl(\{t_e\}\cup\{q_{P,e} : P\in R,\ e\in\supp(P)\}\Bigr).
\]
Each anchor belongs to exactly one $P$; base points belong to none. The constraints come in
two kinds.

\smallskip
\emph{Side constraints.} In coordinate $i$, every generating point of $C_e$ (its base point
and all its anchors) must satisfy: value $\le -1$ if $e\in A_i^*$; value $\ge +1$ if
$e\in B_i^*$; no constraint if $e\notin A_i^*\cup B_i^*$.

\smallskip
\emph{Witness equations.} For each $P=(A',B')\in R$, writing averages componentwise,
\[
w_P \;:=\; \frac{1}{|A'|}\sum_{a\in A'} q_{P,a} \;=\; \frac{1}{|B'|}\sum_{b\in B'} q_{P,b}.
\]
Both components of a Radon partition are nonempty by (A3), so the averages are defined. Since
distinct Radon partitions share no anchors, the witness equations decompose into one scalar
equation per pair $(i,P)$, and the side constraints are componentwise; the system therefore
splits coordinate by coordinate.

\smallskip\noindent\emph{Feasibility in one coordinate.}
Fix $i\le k$ and $P=(A',B')\in R$. On a component $X\in\{A',B'\}$ declare an element
negative, positive, or free according as it lies in $A_i^*$, in $B_i^*$, or in neither side
of the maximal separation $M_i=(A_i^*,B_i^*)$. Lemma~\ref{lem:onecoordinate} gives the
achievable average set of each component: $(-\infty,-1]$ for a negative-pure component,
$[1,\infty)$ for a positive-pure component, and all of $\R$ for a mixed component. The two
achievable sets can be disjoint only when one component is negative-pure and the other is
positive-pure. That would mean
$(A',B')\sqsubseteq(A_i^*,B_i^*)$ or
$(A',B')\sqsubseteq(B_i^*,A_i^*)$, impossible because $M_i$ is a separation and
$(A',B')$ is a Radon partition. Hence the achievable sets intersect.

Choose a common target
\[
\tau_{i,P}=
\begin{cases}
-1 & \text{if some component is negative-pure},\\
+1 & \text{if some component is positive-pure},\\
0 & \text{otherwise.}
\end{cases}
\]
The disjoint pure cases have just been excluded, so $\tau_{i,P}$ is well defined and lies in
both achievable sets. Lemma~\ref{lem:onecoordinate} realizes this target on both components
with rational values of denominator at most $n$ and absolute value at most $2n+1$; multiplying
by the harmless factor $n$ gives the stated numerator bound $n(2n+1)$. Finally set the base
points: $t_e[i]=-1$ if $e\in A_i^*$, $t_e[i]=+1$ if $e\in B_i^*$, and $t_e[i]=0$ otherwise;
base points occur in no equation.

\smallskip\noindent\emph{Verification.}
First, separations. Let $(A,B)\in\apart$. If a component is empty,
$\hull(\emptyset)=\emptyset$ settles the pair with infinite margin. Otherwise $(A,B)$ is
dominated, possibly after a swap, by some $M_i=(A_i^*,B_i^*)$; in coordinate $i$ every
generating point of every $C_a$ with $a\in A\subseteq A_i^*$ has value $\le-1$, hence
$\hull_{\mathcal C}(A)\subseteq\{x:x_i\le-1\}$, and likewise
$\hull_{\mathcal C}(B)\subseteq\{x:x_i\ge+1\}$. The hulls lie in parallel half-spaces at
distance $2$, so $(A,B)\in\apart_{\mathcal C}$ with
$\marg_{\mathcal C}(A,B)\ge 2$.

Second, Radon partitions. Let $(A,B)\in\cross_\Sigma$; it dominates a minimal Radon
partition $P=(A',B')$ with $(A',B')\sqsubseteq(A,B)$, possibly after a swap, which (A1)
absorbs. The witness $w_P$ is, in every coordinate, the common value $\tau_{i,P}$ of the two
averages, hence a single well-defined point of $\R^k$ lying in
$\operatorname{conv}\{q_{P,a}:a\in A'\}\subseteq\hull_{\mathcal C}(A')\subseteq
\hull_{\mathcal C}(A)$ and symmetrically in $\hull_{\mathcal C}(B)$. So
$(A,B)\in\cross_{\mathcal C}$.

Every pair of $\dy{E}$ is a separation or a Radon partition of $\Sigma$, and we have matched
both kinds; hence $\apart_{\mathcal C}=\apart_\Sigma$.
\end{proof}

\begin{remark}\label{rem:dimension}
A realization in dimension $n-1$ is available by a different construction
\cite{Strausz03,BS06}, and $n-1$ is in general far smaller than $k$: for the structure with
no Radon partitions on $n$ sites, $k=2^{n-1}-1$, while $n-1$ suffices and, by
Proposition~\ref{prop:simploid} below, is optimal. The trade is deliberate. The construction
above buys three things for its dimensions: each ambient coordinate is indexed by a maximal
separation, so that, after Lemma~\ref{lem:generation} and Theorem~\ref{thm:subsumption},
each coordinate is literally a deduction and each separating hyperplane is an axis
hyperplane named by the maximal consequence it certifies; the coordinates are rationals of
explicitly bounded height; and the margins are uniform, which is what powers the stability
theory of the next subsection. No margin bounds are asserted for the dimension-$(n-1)$
construction in \cite{Strausz03,BS06}.
\end{remark}

\begin{proposition}[Position relative to the Strausz-Bracho representation]\label{prop:position}
The representation content used in this paper separates into four parts.
\begin{itemize}
\item The axiomatics are the acyclic separoid axioms of \cite{Strausz03} written in
separation polarity; this is a dictionary, not a new class of structures.
\item Dimension-$n-1$ realizability is imported from \cite{Strausz03,BS06} and is used only
in Theorem~\ref{thm:hierarchy} to identify the exact stabilization threshold.
\item The coordinate-by-maximal-separation construction of Theorem~\ref{thm:representation}
is a different, generally higher-dimensional realization. Its dimensions are indexed by
proof certificates rather than optimized geometrically.
\item The rational height bound, the uniform margin $2$, the outer-parallel stability
radius, and the explicit enumeration statement of Proposition~\ref{prop:effective} are the
quantitative additions supplied here.
\end{itemize}
Thus the paper uses known separoid representation theory to locate the optimal ambient
threshold, and supplies a separate rational-margin representation to make the logical and
stability arguments inspectable coordinate by coordinate.
\end{proposition}

\begin{corollary}[Geometric completeness and invariance]\label{cor:invariance}
A relation $\apart\subseteq\dy{E}$ is the apartness relation of some scene if and only if it
satisfies \textup{(A1)--(A3)}, if and only if it is the apartness relation of a scene of
rational polytopes. Consequently $\modsep$ and $\modgeo$ coincide: for all
$\Gamma\cup\{\varphi\}\subseteq\Lang_E$,
\[
\Gamma\modgeo\varphi\iff\Gamma\modsep\varphi .
\]
The Euclidean apparatus (dimension, coordinates, the bodies themselves) contributes no
consequence beyond the three laws.
\end{corollary}

\begin{proof}
The first equivalences combine Proposition~\ref{prop:threelaws} with
Theorem~\ref{thm:representation}. For the second, the classes of atom-valuations induced by
scenes and by apartness separoids coincide, hence so do the consequence relations they
define.
\end{proof}

\begin{example}[The construction, run]\label{ex:construction}
Run Theorem~\ref{thm:representation} on the abstract relation of Example~\ref{ex:running}:
$E=\{a,b,c\}$, maximal bilateral separations $M_1=(\{b\},\{a,c\})$ and $M_2=(\{c\},\{a,b\})$,
unique minimal Radon partition $P=(\{a\},\{b,c\})$. So $k=2$, $m=1$; the points carry
coordinates $(x_1,x_2)$.

Coordinate $1$ (sides: $b\mapsto -$; $a,c\mapsto +$). The component $\{a\}$ of $P$ is
positive-pure; $\{b,c\}$ is mixed. So $\tau_{1,P}=+1$: set $q_{P,a}[1]=1$, and on the mixed
side $q_{P,b}[1]=-1$, $q_{P,c}[1]=3$, whose average is $1$. Bases: $t_a[1]=1$,
$t_b[1]=-1$, $t_c[1]=1$.

Coordinate $2$ (sides: $c\mapsto -$; $a,b\mapsto +$). Again $\{a\}$ is positive-pure and
$\{b,c\}$ mixed: $\tau_{2,P}=+1$, $q_{P,a}[2]=1$, $q_{P,b}[2]=3$, $q_{P,c}[2]=-1$;
$t_a[2]=1$, $t_b[2]=1$, $t_c[2]=-1$.

The bodies come out as
\[
C_a=\{(1,1)\},\qquad
C_b=\operatorname{conv}\{(-1,1),(-1,3)\},\qquad
C_c=\operatorname{conv}\{(1,-1),(3,-1)\},
\]
a point, a vertical segment on $x_1=-1$, and a horizontal segment on $x_2=-1$
(Figure~\ref{fig:construction}). The hyperplane $x_1=0$ certifies $M_1$ and the hyperplane
$x_2=0$ certifies $M_2$, each with clearance $1$ on both sides; the witness
\[
w_P=(1,1)=q_{P,a}=\tfrac12\bigl((-1,3)+(3,-1)\bigr)
\]
lies in $C_a$ and on the segment joining $q_{P,b}\in C_b$ to $q_{P,c}\in C_c$, so
$\{a\}\cross\{b,c\}$ exactly as prescribed. Every pair of $\dy{E}$ can be checked against
the picture in a few seconds; integer coordinates of height $3$ suffice here, well inside
the bound of Theorem~\ref{thm:representation}. The bodies are degenerate (a point and two
segments); Theorem~\ref{thm:stability} repairs this without disturbing a single bit of the
relation.
\end{example}

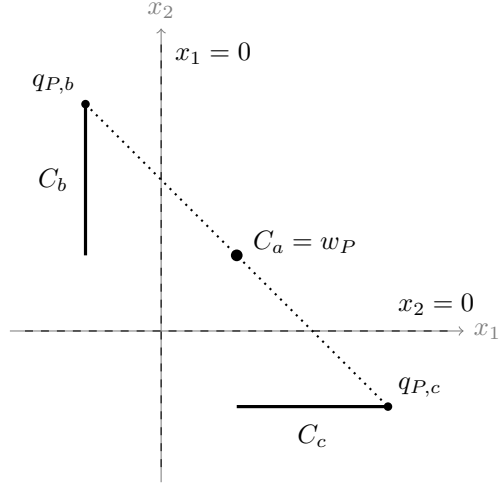
\begin{figure}[t]
\centering
\begin{tikzpicture}[scale=1.0]
\draw[->,black!50] (-2,0) -- (4,0) node[anchor=west] {\small $x_1$};
\draw[->,black!50] (0,-2) -- (0,4) node[anchor=south] {\small $x_2$};
\draw[dashed] (0,-1.8) -- (0,3.8);
\draw[dashed] (-1.8,0) -- (3.8,0);
\draw[dotted,thick] (-1,3) -- (3,-1);
\draw[very thick] (-1,1) -- (-1,3);
\draw[very thick] (1,-1) -- (3,-1);
\fill (1,1) circle (2.2pt);
\node[anchor=west] at (1.08,1.18) {\small $C_a=w_P$};
\node[anchor=east] at (-1.1,2) {\small $C_b$};
\node[anchor=north] at (2,-1.1) {\small $C_c$};
\fill (-1,3) circle (1.6pt); \node[anchor=south east] at (-1,3) {\small $q_{P,b}$};
\fill (3,-1) circle (1.6pt); \node[anchor=south west] at (3,-1) {\small $q_{P,c}$};
\node[anchor=south west] at (0.05,3.4) {\small $x_1=0$};
\node[anchor=south west] at (3.0,0.08) {\small $x_2=0$};
\end{tikzpicture}
\caption{The scene produced in Example~\ref{ex:construction}. The axis hyperplanes $x_1=0$
and $x_2=0$ certify the two maximal separations; the dotted chord from $q_{P,b}$ to
$q_{P,c}$ passes through the witness $w_P=(1,1)$, which is the whole of $C_a$, certifying
the unique minimal Radon partition.}
\label{fig:construction}
\end{figure}

\begin{example}[A Radon lattice from a negative presentation]\label{ex:radonlattice}
Let $E=\{a,b,c,d\}$ and prescribe two minimal Radon partitions,
\[
P_1=(\{a\},\{b,c\}),\qquad P_2=(\{d\},\{b,c\}),
\]
together with their swaps. Let crossing be the upward closure of these two pairs and let
apartness be the complement in $\dy{E}$. This is an apartness separoid by construction: the
crossing relation is symmetric and upward closed, so its complement is symmetric, downward
closed, and contains all vacuous pairs. The two minimal Radon partitions are incomparable,
and their upper cones meet at $(\{a,d\},\{b,c\})$; hence the Radon side is already a small
lattice rather than a list of isolated witnesses. The maximal bilateral separations are
\[
\{a,b,d\}\apart\{c\},\qquad
\{a,c\}\apart\{b,d\},\qquad
\{a,b\}\apart\{c,d\},\qquad
\{a,c,d\}\apart\{b\},
\]
up to symmetry. Theorem~\ref{thm:representation} therefore realizes this four-site example
in four proof-readable coordinates, with one coordinate per displayed maximal separation and
one anchor family per displayed minimal Radon partition.
\end{example}

\begin{example}[Why the proof-readable dimension can be large]\label{ex:largek}
For the $n$-site simploid, every bilateral disjoint pair is apart and there are no Radon
partitions. The maximal bilateral separations are exactly the unordered bipartitions
$A\mid E\setminus A$ with $\emptyset\ne A\ne E$, so
\[
|M_\Sigma|=2^{n-1}-1.
\]
The coordinate construction therefore uses $2^{n-1}-1$ coordinates. By contrast, the
classical geometric realization places the $n$ sites at the vertices of an $(n-1)$-simplex,
and Proposition~\ref{prop:simploid} shows that dimension $n-1$ is optimal. This is the clean
case where $k\gg n$: the high-dimensional realization is intentionally a certificate space,
not a dimension-minimal scene.
\end{example}

\subsection{Margins, distance functions, and stability}\label{sec:margins}

The margin of Definition~\ref{def:margin} organizes the quantitative side of the theory. It
is most cleanly handled through distance functions, whose elementary properties we prove and
whose analytic pedigree we record.

\begin{lemma}[Distance functions]\label{lem:distance}
For a nonempty compact convex $K\subseteq\R^d$ let $d_K(x)=\min_{y\in K}\lVert x-y\rVert$.
Then:
\begin{itemize}
\item[\textup{(i)}] $d_K$ is convex, $1$-Lipschitz, and is the pointwise largest
$1$-Lipschitz function vanishing on $K$;
\item[\textup{(ii)}] if $K\subseteq K'$ then $d_{K'}\le d_K$ pointwise;
\item[\textup{(iii)}] for nonempty compact convex $K,L$,
$\ \marg(K,L)=\min_{x\in\R^d}\bigl(d_K(x)+d_L(x)\bigr)$;
\item[\textup{(iv)}] for $r\ge0$, the sublevel set $\{d_K\le r\}$ equals the outer parallel
body $K+rB^d$, where $B^d$ is the closed unit ball.
\end{itemize}
\end{lemma}

\begin{proof}
(i) Convexity: for $x_0,x_1$ pick nearest points $y_0,y_1\in K$; then
$y_t=(1-t)y_0+ty_1\in K$ and
$d_K(x_t)\le\lVert x_t-y_t\rVert\le(1-t)\lVert x_0-y_0\rVert+t\lVert x_1-y_1\rVert$.
The Lipschitz bound is the triangle inequality through a nearest point. Maximality: if $u$
is $1$-Lipschitz and vanishes on $K$, then for any $x$ and any $y\in K$,
$u(x)\le u(y)+\lVert x-y\rVert=\lVert x-y\rVert$, and minimizing over $y$ gives
$u\le d_K$. (ii) The minimum over a larger set is no larger. (iii) For any $x$, picking
nearest points $p\in K$, $q\in L$ gives
$d_K(x)+d_L(x)=\lVert x-p\rVert+\lVert x-q\rVert\ge\lVert p-q\rVert\ge\marg(K,L)$; at the
midpoint of a closest pair the value $\marg(K,L)$ is attained. (iv) $d_K(x)\le r$ if and
only if $x=y+z$ with $y\in K$ and $\lVert z\rVert\le r$.
\end{proof}

\begin{remark}[Distance functions and the eikonal context]\label{rem:eikonal}
On the open complement of $K$, the function $d_K$ is the viscosity solution of the eikonal
equation $\lvert\nabla u\rvert=1$ with zero boundary data on $\partial K$, in the framework
of \cite{CL83}; see \cite[Ch.~II]{BCD97} for the distance function as the value function of
the associated minimum-time problem. The present paper uses only the convex-analytic facts
proved in Lemma~\ref{lem:distance}: Lipschitz maximality, monotonicity under inclusion, the
minimum formula for the margin, and the identification of sublevel sets with outer parallel
bodies. The eikonal language supplies an interpretation of this bookkeeping, while the
representation, stability, and consequence theorems rest on the elementary distance
statements above.
\end{remark}

\begin{lemma}[The three laws as thresholded monotonicity]\label{lem:threshold}
For every scene $\mathcal C$ and all $(A,B),(A',B')\in\dy{E}$:
\begin{itemize}
\item[\textup{(i)}] $A\apart_{\mathcal C}B\iff\marg_{\mathcal C}(A,B)>0$;
\item[\textup{(ii)}] $\marg_{\mathcal C}(A,B)=\marg_{\mathcal C}(B,A)$;
\item[\textup{(iii)}] if $(A',B')\sqsubseteq(A,B)$ then
$\marg_{\mathcal C}(A',B')\ge\marg_{\mathcal C}(A,B)$;
\item[\textup{(iv)}] $\marg_{\mathcal C}(\emptyset,B)=+\infty$.
\end{itemize}
Thus \textup{(A1)--(A3)} are exactly the thresholded forms, under $\marg>0$, of the
symmetry, antitonicity, and vacuity of the margin; in particular \textup{(A2)} follows from
the comparison in Lemma~\ref{lem:distance}\textup{(ii)}, applied to the nested hulls.
\end{lemma}

\begin{proof}
(i) is Lemma~\ref{lem:separation}; (ii) is symmetry of $\dist$; (iv) is the convention of
Definition~\ref{def:margin}, forced by $\hull(\emptyset)=\emptyset$. For (iii), if either
side of $(A',B')$ is empty the margin is $+\infty$; otherwise
$\hull(A')\subseteq\hull(A)$ and $\hull(B')\subseteq\hull(B)$, so by
Lemma~\ref{lem:distance}(ii) $d_{\hull(A)}\le d_{\hull(A')}$ and likewise for $B$, and the
minimum in Lemma~\ref{lem:distance}(iii) can only increase.
\end{proof}

\begin{theorem}[Uniform margin and stability of the realization]\label{thm:stability}
Let $\Sigma$ be an apartness separoid and let $\mathcal C=(k,(C_e)_e)$ be the scene of
Theorem~\ref{thm:representation}. For $0\le r$ write
$\mathcal C^{(r)}=(k,(C_e+rB^k)_e)$ for the outer parallel scene. Then:
\begin{itemize}
\item[\textup{(i)}] for every $0\le r<1$, the scene $\mathcal C^{(r)}$ realizes $\Sigma$,
with $\marg_{\mathcal C^{(r)}}(A,B)\ge 2-2r$ on every bilateral separation;
\item[\textup{(ii)}] for $0<r<1$ the bodies of $\mathcal C^{(r)}$ are full-dimensional, so
every apartness separoid is realized by a scene of full-dimensional bodies;
\item[\textup{(iii)}] if $\mathcal D=(k,(D_e)_e)$ is any scene with Hausdorff distance
$d_H(D_e,C_e)\le r<1$ for all $e\in E$, then
$\apart_\Sigma\subseteq\apart_{\mathcal D}$: every separation of $\Sigma$ persists in
$\mathcal D$, with margin at least $2-2r$.
\end{itemize}
\end{theorem}

\begin{proof}
Two identities first. For any $X\subseteq\R^k$ and $r\ge0$,
$\operatorname{conv}(X+rB^k)=\operatorname{conv}(X)+rB^k$: the right side is convex and
contains $X+rB^k$, giving one inclusion; conversely a point
$\sum\lambda_j x_j+rb$ equals $\sum\lambda_j(x_j+rb)$, giving the other. Since also
$\bigcup_{a\in A}(C_a+rB^k)=\bigl(\bigcup_{a\in A}C_a\bigr)+rB^k$, we get
\begin{equation}\label{eq:parallelhull}
\hull_{\mathcal C^{(r)}}(A)=\hull_{\mathcal C}(A)+rB^k
\qquad\text{for all }A\subseteq E .
\end{equation}
Second, for nonempty compact convex $K,L$ and $r\ge 0$,
$\dist(K+rB^k,L+rB^k)\ge\dist(K,L)-2r$, since
$\lVert(x+rb_1)-(y+rb_2)\rVert\ge\lVert x-y\rVert-2r$.

(i) Bilateral separations of $\Sigma$ have margin $\ge2$ in $\mathcal C$ by
Theorem~\ref{thm:representation}; by \eqref{eq:parallelhull} and the displacement bound
their margin in $\mathcal C^{(r)}$ is at least $2-2r>0$, so they remain separations.
Vacuous pairs are apart in every scene. Radon partitions of $\Sigma$ cross in
$\mathcal C$; by \eqref{eq:parallelhull} the hulls of $\mathcal C^{(r)}$ contain those of
$\mathcal C$, so the same witness point exhibits crossing. Hence
$\apart_{\mathcal C^{(r)}}=\apart_\Sigma$.

(ii) $C_e+rB^k$ contains a ball of radius $r$ around any point of $C_e$.

(iii) $d_H(D_e,C_e)\le r$ gives $D_e\subseteq C_e+rB^k$, hence, by the argument for
\eqref{eq:parallelhull} applied to unions,
$\hull_{\mathcal D}(A)\subseteq\hull_{\mathcal C}(A)+rB^k$ for every $A$. For a bilateral
separation $(A,B)$ of $\Sigma$, the displacement bound gives
$\marg_{\mathcal D}(A,B)\ge\marg_{\mathcal C}(A,B)-2r\ge 2-2r>0$.
\end{proof}

\begin{remark}[Smoothing, and the limits of stability]\label{rem:smoothing}
For $0<r<1$ each body of $\mathcal C^{(r)}$ is an outer parallel body at positive radius and
therefore has $C^{1,1}$ boundary: such a body has reach at least $r$ in the sense of
\cite{Federer59}, and the regularity of parallel bodies of convex sets is classical
\cite[\S3.1]{Schneider14}. Every apartness separoid is thus realized by full-dimensional
bodies with $C^{1,1}$ boundaries, the centers of the construction remaining rational. The
asymmetry in Theorem~\ref{thm:stability}(iii) is intrinsic and not an artefact of the proof:
separations carry a margin and survive every sufficiently small perturbation, while a
crossing realized by a single touching point can be destroyed by an arbitrarily small one.
Within the outer parallel family both persist, which is why
Theorem~\ref{thm:stability}(i) is an equality of relations while (iii) is, and must be, an
inclusion. The relation produced by Theorem~\ref{thm:representation} therefore sits in
equilibrium: enlargement by any radius below $1$ moves every body and no bit.
\end{remark}

\subsection{Two exact computations}\label{sec:exact}

\begin{definition}\label{def:simploid}
The \emph{$d$-dimensional simploid} $\sigma^d$ is the apartness separoid on a set of $d+1$
sites in which every disjoint pair is apart; equivalently, the acyclic separoid of order
$d+1$ with no Radon partitions \cite{Strausz03}.
\end{definition}

\begin{proposition}\label{prop:simploid}
$\gd(\sigma^d)=d$ for every $d\ge 0$.
\end{proposition}

\begin{proof}
\emph{Upper bound.} Realize the $d+1$ sites by the vertices $v_0,\dots,v_d$ of a
$d$-simplex in $\R^d$, each body a single vertex. If disjoint $A,B$ had
$x\in\hull(A)\cap\hull(B)$, then writing $x$ as a convex combination over $A$ and over $B$
and subtracting yields a nontrivial affine dependence among the affinely independent points
$v_0,\dots,v_d$: the supports of the two combinations are disjoint, and at least one
coefficient on each side is nonzero unless a combination is empty, which vacuity of the
hull excludes. Hence every disjoint pair is apart.

\emph{Lower bound.} Suppose a scene $(C_e)_{e\in E}$ with $|E|=d+1$ in $\R^{d-1}$ realized
$\sigma^d$ for some $d\ge1$. Choose a point $p_e\in C_e$ for each $e$; this is the choice
construction of \cite{Strausz03}. The $d+1=(d-1)+2$ points $p_e$ in $\R^{d-1}$ admit, by
Radon's theorem \cite{Radon21}, a partition $E=A\,\dot\cup\,B$ into nonempty parts with
$\operatorname{conv}\{p_a:a\in A\}\cap\operatorname{conv}\{p_b:b\in B\}\neq\emptyset$.
Since $p_e\in C_e$, these point hulls are contained in $\hull(A)$ and $\hull(B)$
respectively, so $(A,B)$ would be a Radon partition of the scene; but $\sigma^d$ has none.
For $d=0$ the claim is trivial.
\end{proof}

\begin{theorem}[The dimension hierarchy and its threshold]\label{thm:hierarchy}
Let $n=|E|$. For $d\ge0$ and a subset $S\subseteq E$ with $|S|=d+2$, let
\[
\chi_S \;=\; \bigvee_{\substack{S=A\,\dot\cup\,B\\ A,B\neq\emptyset}} \neg\,\atomf{A}{B},
\]
the disjunction running over unordered partitions of $S$ into two nonempty parts. Then:
\begin{itemize}
\item[\textup{(i)}] $\models_{d+1}\ \subseteq\ \models_{d}$ for every $d\ge0$, and
$\modgeo=\bigcap_{d\ge0}\models_d$;
\item[\textup{(ii)}] for every $0\le d\le n-2$ and every $(d{+}2)$-subset $S\subseteq E$,
the formula $\chi_S$ is valid over scenes in $\R^d$ and refuted by a scene in $\R^{d+1}$;
hence $\models_d\ \supsetneq\ \models_{d+1}$;
\item[\textup{(iii)}] $\models_d\ =\ \modgeo$ for every $d\ge n-1$.
\end{itemize}
Consequently the chain
$\models_0\supsetneq\models_1\supsetneq\cdots\supsetneq\models_{n-2}\supsetneq
\models_{n-1}=\models_{n}=\cdots=\modgeo$
is strict at every step below $n-1$ and constant from $n-1$ onward; the stabilization
threshold is exactly $n-1$.
\end{theorem}

\begin{proof}
(i) An isometric embedding $\R^d\hookrightarrow\R^{d+1}$ carries a scene in $\R^d$ to a
scene in $\R^{d+1}$ whose bodies lie in a hyperplane; convex hulls of subsets of a
hyperplane are computed within it, so the apartness relation is preserved and every
$d$-dimensional countermodel is a $(d{+}1)$-dimensional one. The class of all scenes is the
union over $d$ of the fixed-dimension classes, which gives the intersection formula.

(ii) Validity over $\R^d$: in any scene, choose $p_e\in C_e$ for $e\in S$; these are
$d+2$ points of $\R^d$, so Radon's theorem \cite{Radon21} provides a partition
$S=A\,\dot\cup\,B$ into nonempty parts whose point hulls intersect; the point hulls are
contained in $\hull(A)$ and $\hull(B)$, so $\{A,B\}$ crosses and the corresponding disjunct
of $\chi_S$ holds. Refutation in $\R^{d+1}$: place the sites of $S$ at the vertices of a
$(d{+}1)$-simplex and the remaining sites of $E$ anywhere; by the upper-bound argument of
Proposition~\ref{prop:simploid} every bilateral pair inside $S$ is apart, so every disjunct
of $\chi_S$ fails. With $0\le d\le n-2$ such a subset $S$ exists, and combining with (i)
gives strictness.

(iii) Every apartness separoid over $E$ is realizable in $\R^{\,n-1}$ by the representation
theorem of \cite{Strausz03,BS06}; hence every abstract countermodel to a consequence is
witnessed in dimension $n-1$, so $\models_{n-1}\subseteq\modsep=\modgeo$, the equality by
Corollary~\ref{cor:invariance}. With $\modgeo\subseteq\models_d$ for all $d$ and the chain
of (i), equality holds for all $d\ge n-1$. This is the single point at which a result is
imported; Theorem~\ref{thm:representation} independently yields stabilization from the
finite threshold $\max_\Sigma |M_\Sigma|$, and the imported bound pins the threshold at
$n-1$, which (ii) shows cannot be lowered.
\end{proof}

For $d=1$ and $S=\{a,b,c\}$ the formula $\chi_S$ reads
\begin{multline*}
\neg\atomf{\{a\}}{\{b,c\}}\ \vee\ \neg\atomf{\{b\}}{\{a,c\}}\ \vee\
\neg\atomf{\{c\}}{\{a,b\}}\\
\vee\ \neg\atomf{\{a\}}{\{b\}}\ \vee\
\neg\atomf{\{a\}}{\{c\}}\ \vee\ \neg\atomf{\{b\}}{\{c\}},
\end{multline*}
valid on a line, where three pairwise disjoint intervals are ordered and the middle one is
swallowed by the span of the outer two, and refuted by a nondegenerate triangle of points in
the plane. Theorem~\ref{thm:hierarchy} delimits the scope of everything that follows: the
calculus of the next section axiomatizes the dimension-free relation
$\modgeo=\modsep$, the stable core that no choice of ambient space can disturb, and it
identifies the exact moment at which the ambient dimensions stop being individually
visible to the language.

\section{The calculus}\label{sec:calculus}

A \emph{positive sequent} is an expression $\Gamma\vdz\varphi$ with
$\Gamma\cup\{\varphi\}\subseteq\Atoms_E$.

\begin{definition}[The subsumption calculus $\mathsf{SC}_0$]\label{def:sc0}
The positive sequents are generated by:
\[
\frac{}{\Gamma,\ \atomf{A}{B}\ \vdz\ \atomf{A}{B}}\ (\mathrm{id})
\qquad\qquad
\frac{}{\Gamma\ \vdz\ \atomf{A}{B}}\ (\mathrm{triv})\quad
\text{\small provided $A=\emptyset$ or $B=\emptyset$}
\]
\[
\frac{\Gamma\ \vdz\ \atomf{A}{B}}{\Gamma\ \vdz\ \atomf{B}{A}}\ (\mathrm{sym})
\qquad\qquad
\frac{\Gamma\ \vdz\ \atomf{A}{B}}{\Gamma\ \vdz\ \atomf{A'}{B'}}\ (\mathrm{sub})\quad
\text{\small provided $A'\subseteq A$ and $B'\subseteq B$.}
\]
\end{definition}

The rule (sub) is the proof-theoretic face of (A2): geometric weakening. Note the inversion
relative to logical weakening: shrinking the asserted sets weakens the geometric claim, and,
by Lemma~\ref{lem:threshold}(iii), increases its margin.

\begin{theorem}[Subsumption form of positive entailment]\label{thm:subsumption}
For $\Gamma\cup\{\atomf{A}{B}\}\subseteq\Atoms_E$ the following are equivalent:
\begin{itemize}
\item[\textup{(i)}] $\Gamma\vdz\atomf{A}{B}$;
\item[\textup{(ii)}] $\Gamma\modsep\atomf{A}{B}$;
\item[\textup{(iii)}] $\Gamma\modgeo\atomf{A}{B}$;
\item[\textup{(iv)}] $A=\emptyset$, or $B=\emptyset$, or some $\atomf{A'}{B'}\in\Gamma$ has
$(A,B)\sqsubseteq(A',B')$ or $(A,B)\sqsubseteq(B',A')$.
\end{itemize}
Moreover every derivable positive sequent has a derivation consisting of one instance of
\textup{(id)} or \textup{(triv)} followed by at most one instance of \textup{(sym)} and then
at most one instance of \textup{(sub)}.
\end{theorem}

\begin{proof}
(iv)$\Rightarrow$(i) and the normal form: vacuous targets are (triv) axioms; a direct
domination is (id) followed by one (sub); a swapped domination is (id), one (sym), one
(sub). (i)$\Rightarrow$(ii): each rule is sound over apartness separoids: (triv) by (A3),
(sym) by (A1), (sub) by (A2), (id) trivially. (ii)$\Rightarrow$(iv): let
$\Gamma^\circ=\{(A',B'):\atomf{A'}{B'}\in\Gamma\}$ and consider the \emph{minimal model}
$\Sigma_\Gamma$ with $\apart_{\Sigma_\Gamma}=\cl(\Gamma^\circ)$, an apartness separoid by
Lemma~\ref{lem:generation} satisfying every atom of $\Gamma$. If (iv) fails, then by
Definition~\ref{def:closure} the pair $(A,B)$ is not in $\cl(\Gamma^\circ)$, so
$\Sigma_\Gamma\not\models\atomf{A}{B}$, refuting (ii).
(ii)$\Leftrightarrow$(iii) is Corollary~\ref{cor:invariance}.
\end{proof}

\begin{corollary}[Cut admissibility]\label{cor:cut}
If $\Gamma\vdz\delta$ and $\Gamma\cup\{\delta\}\vdz\varphi$, then $\Gamma\vdz\varphi$.
\end{corollary}

\begin{proof}
Apply criterion (iv) to the second sequent. If $\varphi$ is vacuous or dominated within
$\Gamma$, the conclusion is immediate. If $\varphi$ is dominated by
$\delta=\atomf{D}{D'}$, two cases remain. When $\delta$ is itself dominated by some member
of $\Gamma$, compose the dominations: inclusion is transitive and a double swap is the
identity, so the four direct or swapped combinations reduce to a direct or a swapped
domination by that member. When $\delta$ is vacuous, say $D=\emptyset$, a direct domination
forces the first component of $\varphi$ to be empty and a swapped one forces its second
component to be empty, so $\varphi$ is vacuous. In all cases (iv) holds for
$\Gamma\vdz\varphi$.
\end{proof}

\begin{remark}[Interpolation trivializes]\label{rem:interpolation}
If $\atomf{A}{B}\vdz\atomf{A'}{B'}$ and the conclusion is not vacuous, then by criterion
(iv) $A'\cup B'\subseteq A\cup B$: the conclusion mentions only sites of the premise and is
its own interpolant. The calculus has no transversal content to interpolate away; all
content is containment.
\end{remark}

\begin{definition}[The Boolean calculus $\mathsf{AC}_E$]\label{def:ac}
$\mathsf{AC}_E$ is any standard complete Hilbert system for classical propositional logic
over the atoms $\Atoms_E$, extended by the axiom schemes, ranging over $\dy{E}$:
\[
(\mathrm T)\ \ \atomf{A}{B}\ \ \text{\small for $A=\emptyset$ or $B=\emptyset$};\qquad
(\mathrm S)\ \ \atomf{A}{B}\rightarrow\atomf{B}{A};
\]
\[
(\mathrm W)\ \ \atomf{A}{B}\rightarrow\atomf{A'}{B'}\ \ \text{\small for
$A'\subseteq A$, $B'\subseteq B$}.
\]
Derivability from a set $\Gamma\subseteq\Lang_E$ is written $\Gamma\vdash\varphi$.
\end{definition}

\begin{theorem}[Soundness, completeness, decidability]\label{thm:completeness}
For all $\Gamma\cup\{\varphi\}\subseteq\Lang_E$:
\[
\Gamma\vdash\varphi
\iff \Gamma\modsep\varphi
\iff \Gamma\modgeo\varphi,
\]
and these relations are decidable.
\end{theorem}

\begin{proof}
A Boolean valuation $v$ of $\Atoms_E$ is the same thing as a relation
$\apart_v\subseteq\dy{E}$, and $v$ satisfies all instances of (T), (S), (W) precisely when
$\apart_v$ satisfies (A3), (A1), (A2): for (S) note that the instances at $(A,B)$ and at
$(B,A)$ jointly give the biconditional. Hence the models of the scheme set, among Boolean
valuations, are exactly the apartness separoids. Soundness of $\vdash$ for $\modsep$
follows, the propositional part being classically sound and the schemes valid by
Proposition~\ref{prop:threelaws} read abstractly. For completeness, suppose
$\Gamma\modsep\varphi$. The atom set $\Atoms_E$ is finite, so by completeness of the
propositional base, $\varphi$ is derivable from $\Gamma$ together with the finitely many
scheme instances if and only if every Boolean valuation satisfying both satisfies
$\varphi$; and the valuations satisfying the scheme instances are exactly the apartness
separoids, over which $\Gamma\models\varphi$ holds by hypothesis. The second equivalence is
Corollary~\ref{cor:invariance}. Decidability follows from finiteness of the valuation
space; the sharp account is Theorem~\ref{thm:np}.
\end{proof}

\begin{remark}[Where the content lies]\label{rem:content}
The first equivalence of Theorem~\ref{thm:completeness} is completeness relative to a
finite-alphabet propositional base and is, as such, routine; we state it because the
calculus needs a name. The content of the theorem is the identification of the models of
the scheme set with the apartness separoids, and the second equivalence, which rests
entirely on Theorem~\ref{thm:representation}. Stripped of scaffolding, the theorem says:
Euclidean separation talk proves nothing that the three schemes do not already prove.
\end{remark}

\begin{proposition}[The calculus as a syntactic shadow]\label{prop:shadow}
For every positive data set $\Gamma\subseteq\Atoms_E$, the atomic consequences of the Boolean
calculus are exactly the atoms true in the minimal separoid $\Sigma_\Gamma$ generated by
$\Gamma$:
\[
\{\alpha\in\Atoms_E:\Gamma\vdash\alpha\}
=
\{\alpha\in\Atoms_E:\Sigma_\Gamma\models\alpha\}
=
\cl(\Gamma^\circ).
\]
Consequently the logical infrastructure is conservative over the geometric closure operator:
classical propositional reasoning can select among admissible valuations, but it does not
create a new positive apartness fact from positive data.
\end{proposition}

\begin{proof}
The first equality is Theorem~\ref{thm:completeness} restricted to atomic conclusions and
the model $\Sigma_\Gamma$ used in the proof of Theorem~\ref{thm:subsumption}; the second is
Lemma~\ref{lem:generation}. The final sentence is just this equality read syntactically.
\end{proof}

\begin{remark}[Robustness under alphabet extension]\label{rem:alphabet}
Satisfiability and validity do not depend on the ambient alphabet. Let $E\subseteq E'$ and
let $\varphi\in\Lang_E\subseteq\Lang_{E'}$. If a structure over $E'$ satisfies $\varphi$,
its restriction to $\dy{E}$ is an apartness separoid over $E$ assigning the same values to
all atoms of $\varphi$. Conversely, if $\Sigma$ over $E$ satisfies $\varphi$, then
$\cl_{E'}(\apart_\Sigma)$ is a structure over $E'$ whose restriction to $\dy{E}$ is exactly
$\apart_\Sigma$: a pair over $E$ is dominated by a pair over $E$, possibly after a swap, in
$\dy{E'}$ if and only if it is so dominated in $\dy{E}$. Hence the decision problems of the
next section are well posed without fixing $E$ in advance.
\end{remark}

\section{Complexity}\label{sec:complexity}

We adopt the explicit encoding: a formula of $\Lang_E$ lists each atom with its two
components written out as sets of site identifiers; the input size $|\varphi|$ is the total
length. The key combinatorial fact is that local consistency of an atom valuation is global
realizability.

\begin{lemma}[Extension]\label{lem:extension}
Let $T\subseteq\Atoms_E$ be finite and $v:T\to\{0,1\}$. There is an apartness separoid
(hence, by Theorem~\ref{thm:representation}, a scene of rational polytopes) agreeing with
$v$ on $T$ if and only if for every $\beta=\atomf{A}{B}\in T$ with $v(\beta)=0$: the pair
$(A,B)$ is bilateral and is not dominated, directly or after a swap, by any $(A',B')$ with
$\atomf{A'}{B'}\in T$ and $v(\atomf{A'}{B'})=1$. The condition is checkable in time
$O(|T|^2\cdot\ell)$, where $\ell$ bounds the encoding length of an atom.
\end{lemma}

\begin{proof}
Necessity: in any apartness separoid, vacuous atoms are true by (A3) and pairs dominated by
true pairs are true by (A1) and (A2). Sufficiency: let $\Sigma$ have
$\apart_\Sigma=\cl(\{(A',B'):v(\atomf{A'}{B'})=1\})$. By Lemma~\ref{lem:generation} this is
an apartness separoid; atoms set to $1$ are satisfied because $\cl$ contains its
generators; an atom $\beta$ set to $0$ is falsified because, by
Definition~\ref{def:closure}, membership of its pair in the closure would mean exactly
vacuity or domination by a generator, both excluded by the condition. The check compares
each $0$-atom against each $1$-atom with four subset tests on explicitly listed sets.
\end{proof}

\begin{theorem}\label{thm:np}
Satisfiability of formulas of $\Lang_E$ over scenes (equivalently, by
Corollary~\ref{cor:invariance}, over apartness separoids) is NP-complete; validity is
coNP-complete; positive entailment $\Gamma\vdz\varphi$ is decidable in time
$O(|\Gamma|+|\varphi|)$ up to the cost of set comparisons, hence in linear time for sorted
encodings.
\end{theorem}

\begin{proof}
\emph{Membership.} Guess $v$ on the atoms occurring in $\varphi$, verify the Boolean
evaluation, and verify the condition of Lemma~\ref{lem:extension}; all in time polynomial
in $|\varphi|$, and correct by that lemma. Validity is the complement of satisfiability of
the negation. Positive entailment is criterion (iv) of Theorem~\ref{thm:subsumption}: one
vacuity test and one scan of $\Gamma$ with constantly many subset tests per premise.

\emph{Hardness.} Reduce propositional satisfiability. Given a formula $\theta$ over
variables $x_1,\dots,x_m$, take $E_\theta=\{a_1,b_1,\dots,a_m,b_m\}$, fresh and pairwise
distinct, and let $\theta^*$ be $\theta$ with each $x_i$ replaced by the atom
$\alpha_i=\atomf{\{a_i\}}{\{b_i\}}$. The map is computable in linear time, and
Remark~\ref{rem:alphabet} licenses the change of alphabet. If a scene satisfies
$\theta^*$, the truth values of the $\alpha_i$ satisfy $\theta$. Conversely, given a
satisfying assignment $w$ of $\theta$, set $v(\alpha_i)=w(x_i)$; the family
$\{\alpha_i\}$ meets the condition of Lemma~\ref{lem:extension} under \emph{every} $v$,
because all components are nonempty singletons and a domination
$(\{a_j\},\{b_j\})\sqsubseteq(\{a_i\},\{b_i\})$, direct or swapped, forces equality of
singletons across a partitioned alphabet, hence $i=j$. So some apartness separoid, and
therefore some scene of rational polytopes, realizes $v$ and satisfies $\theta^*$. Thus
$\theta$ is satisfiable if and only if $\theta^*$ is, giving NP-hardness; the validity
statement is dual via $\theta\mapsto\neg\theta^*$.
\end{proof}

The asymmetry deserves a sentence. Positive information about apartness composes only by
subsumption and is decided by a scan, while the hardness enters exclusively through Boolean
combination, that is, through the freedom to demand that certain pairs \emph{cross}.
Geometry supplies that freedom wholesale (Lemma~\ref{lem:extension}) and adds nothing else
(Corollary~\ref{cor:invariance}).

\section{Stratification and conservativity}\label{sec:stratification}

The calculus $\mathsf{AC}_E$ was presented in a particular discipline, which we now make
explicit and exploit. Assign \emph{grades}: sites and their sets have grade $0$; atoms have
grade $1$; compound formulas have grade $2$. The discipline is in the spirit of
stratification in logic programming \cite{ABW88}, transposed from recursion through
negation to the relationship between a calculus and its subject matter.

\begin{definition}[Stratified presentation]\label{def:discipline}
A presentation of a calculus over $\Lang_E$ is \emph{stratified} when its objects are assigned
grades and the following three conditions hold. For a scheme $S$, let $\mathrm{concl}(S)$ be
the set of grades of the formulas that $S$ derives, and let $\mathrm{exh}_g(S)$ be the set of
closed objects of grade $g$ that occur as displayed constants in the statement of $S$.
\begin{itemize}
\item[\textup{(D1)}] every axiom and rule is schematic, with all metavariables typed by
 grade;
\item[\textup{(D2)}] every side condition of a scheme constrains only objects whose grades
are strictly below every grade in $\mathrm{concl}(S)$;
\item[\textup{(D3)}] for every $g\in\mathrm{concl}(S)$, $\mathrm{exh}_g(S)=\emptyset$.
\end{itemize}
Thus a scheme may range over the grade it derives, but it may not contain a closed instance
of that grade inside its own statement.
\end{definition}

\begin{proposition}\label{prop:disciplineholds}
The presentation of $\mathsf{AC}_E$ in Definition~\ref{def:ac}, with positive sequents
governed by Definition~\ref{def:sc0}, is stratified.
\end{proposition}

\begin{proof}
All axioms and rules are schemes (D1). Side conditions: (triv) and (T) test emptiness of a
grade-$0$ set; (sym) and (S) perform a transposition of grade-$0$ data; (sub) and (W) test
inclusions between grade-$0$ sets; the propositional rules carry no side conditions. In
each case the constrained material is of grade $0$ while the derived formulas are of grade
$1$ or $2$, giving (D2). No scheme mentions a particular site, a particular atom, or a
particular compound, only schematic letters, so (D3) holds.
\end{proof}

The discipline is not decoration; it has a theorem attached. Stratification forbids the
governing stratum from carrying instances of the governed one inside its own laws, and the
payoff is that the governing stratum cannot \emph{create} members of the governed one
either.

\begin{theorem}[Conservativity: the Boolean stratum is inert]\label{thm:conservativity}
For all $\Gamma\cup\{\varphi\}\subseteq\Atoms_E$:
\[
\Gamma\vdash\varphi \iff \Gamma\vdz\varphi .
\]
\end{theorem}

\begin{proof}
($\Leftarrow$) Each rule of $\mathsf{SC}_0$ is simulated in $\mathsf{AC}_E$: (triv) by (T),
(sym) by (S) and modus ponens, (sub) by (W) and modus ponens, (id) by the premise itself.
($\Rightarrow$) If $\Gamma\vdash\varphi$ then $\Gamma\modsep\varphi$ by
Theorem~\ref{thm:completeness}, hence $\Gamma\vdz\varphi$ by
Theorem~\ref{thm:subsumption}. The proof is short; the content is the discipline it
certifies.
\end{proof}

\begin{corollary}\label{cor:noforcing}
No apartness assertion is derivable from positive apartness data by any amount of Boolean
reasoning, classical negation, disjunction, and case analysis included, unless it is
vacuous or a subsumption instance of a single datum. In particular the minimal model
$\Sigma_\Gamma$ of Theorem~\ref{thm:subsumption} satisfies exactly the
$\vdash$-consequences of $\Gamma$ among atoms.
\end{corollary}

Read together with Section~\ref{sec:margins}, the corollary closes a circle. The stability
layer says that perturbing every body by any radius below the margin changes no bit of the
relation; the syntactic layer says that no reasoning over the bits changes them either. The
geometry realizes every consistent pattern, the closure operator names the only patterns
that are forced, and the stratified calculus is the syntax of exactly that forcing, no
more and, by completeness, no less.

\section{Concluding remarks}\label{sec:conclusion}

Three laws, none of them surprising on its own, turn out to exhaust what disjointness of
joint convex hulls can impose on a finite index set once the ambient dimension is left
free. Theorem~\ref{thm:hierarchy} locates the exact moment at which the dimensions stop
mattering: below $|E|-1$ each additional dimension is genuinely new, strictly enlarging
what can be refuted, and each step is witnessed by a single Radon-type formula; from
$|E|-1$ onward the dimensions repeat one another and the consequence relation, unable to
tell them apart, settles into its limit. What survives the settling is small and exactly
known: a closure operator, a four-rule calculus whose derivations normalize in three
steps, a minimal model. The stability layer of Section~\ref{sec:margins} adds that this
core is carried with slack rather than by accident: every separation in the canonical
realization holds with a fixed margin, every body can be thickened and smoothed without
disturbing a single bit, and the relation sits in an equilibrium that perturbation below
the margin cannot reach. The syntactic layer adds the complementary guarantee: no Boolean
argument, however elaborate, forces a bit that the data did not already contain by
subsumption.

The contrast with the fixed-dimension landscape is sharp. Realizability of point
configurations in a prescribed dimension is governed by universality phenomena
\cite{Mnev88}, and within separoid theory the fixed-dimension invariant $\gd$ already
separates structures that the present language, by design and by
Corollary~\ref{cor:invariance}, identifies; Proposition~\ref{prop:simploid} and
Theorem~\ref{thm:hierarchy} mark the boundary precisely. A reader who came for the wilder
questions will find them exactly where they were left, in the coordinates. The
dimension-free relation itself asks for no rescue from its plainness: it is finitely
axiomatized, effectively and robustly realized, decided at the propositional price, and
conservative over its own positive fragment. The stable core of separation talk is small,
exactly known, and entirely subsumptive, and the paper's claim is that holding the whole
theory, every scene, every margin, every derivation, steadily inside that core is not a
limitation of the language but the precise shape of what the language, on its own, was
ever going to mean.

\end{document}